\documentclass[12pt]{amsart}
\usepackage[utf8]{inputenc}
\usepackage{amsmath,amssymb,amsthm}
\usepackage{url}
\usepackage{geometry}
\usepackage{graphicx}
\usepackage{hyperref}
\usepackage[numbers]{natbib}
\usepackage[dvipsnames,usenames]{xcolor} 
\usepackage{mathtools}
\usepackage{cleveref}
\usepackage{tikz}
\usepackage{comment}
\usepackage{chemfig}
\usepackage[linesnumbered,ruled,vlined]{algorithm2e}
\usepackage{booktabs}
\usepackage{enumerate}

\newtheorem{theorem}{Theorem}
\numberwithin{theorem}{section}
\newtheorem{proposition}[theorem]{Proposition}
\newtheorem{lemma}[theorem]{Lemma}

\theoremstyle{definition}

\newtheorem{remark}[theorem]{Remark}
\newtheorem{example}[theorem]{Example}
\newtheorem{conjecture}[theorem]{Conjecture}

\graphicspath{{Figures/}}

\newcommand{\RR}{\mathbb{R}}
\newcommand{\ZZ}{\mathbb{Z}}
\newcommand{\Rp}{\mathbb{R}_{>0}}

\newcommand{\Rn}{\mathbb{R}_{\geq 0}}

\newcommand{\hx}{(h,x)}
\newcommand{\hxs}{(h^*,x^*)}
\renewcommand{\k}{\kappa}
\renewcommand{\l}{\lambda}
\newcommand{\Skd}{\Sigma_{\kappa,d}}
\newcommand{\Sk}{\Sigma_{\kappa}}
\newcommand{\Shl}{\Sigma_{m,\ell}}

\DeclareMathOperator{\diag}{diag}

\DeclareMathOperator{\sign}{sign}
\DeclareMathOperator{\Id}{Id}
\DeclareMathOperator{\conv}{conv}
\DeclareMathOperator{\np}{N}

\title[Disconnectivity in Multistationarity Regions]{Disconnectivity in Multistationarity Regions of Cascade of Goldbeter--Koshland Loops}
\author{Nidhi Kaihnsa and Kaizhang Wang}

\begin{document}
	
	

	\begin{abstract}
		Dynamics of reaction networks is often modeled by parameterised polynomials and describing the set of parameters for which the system attains multiple positive equilibrium states is a challenging problem. In the full parameter space, determined by the reaction rate constants and the total concentrations, the existing methods can give an upper bound on the number of connected components of regions that enable multistationarity and this can give sufficient condition to establish path connectivity of these regions. In this article, we focus on cascade network of Goldbeter-Koshland loops and show that for this network, with $n\geq 2$ phosphorylation sites and at least one shared phosphatase for dephosphorylation, the region in the space of reaction rates is connected while it may not be true in the full parameter space. We explicitly show that there is a gap between the upper and lower bound for the number of connected regions for the case when $n=2$.
		
		\vskip 0.1in
		
		\noindent
		{\bf Keywords:} phosphorylation networks, connectivity, Newton polytope.

	\end{abstract}
    \maketitle
	
	\section{\bf Introduction}

	In reaction networks theory, parametrised first-order ordinary differential equations (ODE) system is predominantly used to model the interaction of species. Within the framework of mass-action kinetics, the rate of change of concentration of species is described by parametric polynomials \cite{feinbergbook}. One of the main problems of interest focuses on identifying the conditions for these systems to have \textit{multistationarity}. In other words, for which parameters do these polynomials have multiple positive real solutions. Original motivation arises from the need to understand multistability which has been linked to several cellular processes including cell-signalling and switch-like responses \cite{G-distributivity,rendall-MAPK,Switchlike,qiao:oscillations}. Since multistationarity is a necessary condition for multistability, it remains a key object of study.

	There is extensive interest in deciding multistationarity of parametric systems. 
	Once the capacity for multistationarity is established in a reaction network, the next challenging step is to describe the parameter region where it occurs.
	Identifying the entire parameter regions for multistationary is known to be a special real quantifier elimination problem in real algebraic geometry.
	However, the complexity of existing symbolic methods is too high even for small networks \cite{BRADFORD202084}. To study these regions, therefore, it is often useful to instead focus on their topological properties like their shape and connectivity. For reaction networks, lack of connectivity of parameter regions of multistationarity has also been associated with complicated switching behaviours \cite{ParamGeo}.

	Reaction networks, in particular, have two sets of parameters called reaction rates and total concentrations. We denote the space of all these parameters as $\Sigma_{\k,d}$. Ideally, we are interested in identifying the parameters in $\Sigma_{\k,d}$ for which the reaction networks have multiple steady states. However, there is no known direct way of addressing this problem in the space of all parameters $\Sigma_{\k,d}$. Instead, we study this problem in the projected space $\Sk$ or $\Sigma_d$. Finding connected regions in projected space gives a lower bound on the number of connected components in $\Skd$. One of the first instances of studying the topological properties of the regions in $\Sk$ was over the phosphorylation networks
	in \cite{FKWY,FKWY2}, where the authors symbolically established that the parameter region for multistationarity is path connected in the space of reaction rates ($\Sk$) for strongly irreversible phosphorylation networks with arbitrary number of sites available for phosphorylation. This exploration was possible due to the existence of a critical polynomial $q_\k(x)$ (\cite{CC2017}) and then the problem is reduced to studying the positivity of this polynomial which was carried out using various combinatorial and real geometric methods. Further research along this direction include exploiting numerical methods to map out parameter geography \cite{ParamGeo} and determine these regions for small reaction networks \cite{dennis2023connectedness}. For the study of multistationarity over the space of total concentration, $\Sigma_d$, see \cite{CIK}.
	
	The incidence variety of the parametric polynomial system associated with reaction networks can be reparametrised using the convex coordinates introduced by Clarke \cite{Clarke} and denoted by $(m,\ell).$ Consequently, the critical polynomial can also be reparametrised as $q(m,\ell)$. The number of connected regions in $\Shl$ for which the polynomial $q(m,\ell)$ is negative give an upper bound on the number of connected components for multistationarity in $\Skd$ \cite{ConnectivityPaper}. If, in particular, it is connected in $\Shl$, then the parameter region in $\Skd$ is connected. This has been used as a sufficient condition to establish connectivity in $\Skd$ for differrent networks. In \cite{kaihnsa_connectivity_2024}, the authors established connectivity of parameter regions in $\Skd$ for the family of weakly and strongly irreversible phosphorylation networks.  
	
	In this article, we will focus on the the cascades of enzymatic Goldbeter-Koshland loops \cite{AG1981} with $n$ layers (or $n$-layer cascades for short).
	This reaction network plays a crucial role in signal transduction inside cells.
	It takes the following form:
	\begin{equation}\label{eq:nlayernet}\tag{$\mathcal{N}_{n}$}
		Y_i+X_{i-1}  \stackrel{\kappa_{6i-5}}{\underset{\kappa_{6i-4}}{\rightleftarrows}} V_i \stackrel{\kappa_{6i-3}}{\longrightarrow} X_i+X_{i-1}, \quad
		X_i+W_i  \stackrel{\kappa_{6i-2}}{\underset{\kappa_{6i-1}}{\rightleftarrows}} U_i \stackrel{\kappa_{6i}}{\longrightarrow} Y_i+W_i, \quad i=1, \ldots, n .
	\end{equation}
	This network involves two phosphorylation cycles.
	For any $i \in \{1,\ldots,n\}$, $X_i$ and $Y_i$ are the substrate proteins, representing the presence and absence of a phosphate group, respectively.
	In the first layer, the phosphorylation is catalyzed by $X_0$, and in other layers, the process is catalyzed by the activated protein of the previous layer.
	In the $i$-th layer, phosphatase $W_i$ catalyses the dephosphorylation.
	Some of the $W_i$ can be the same species, that is, different layers can be catalyzed by the same phosphatase.
	It has been shown in \cite{EF2012} that if all phosphatases are different, then there cannot be more than one positive steady state.
	Therefore, we will always assume that $n \geq 2$ and consider the case when there exist at least two layers sharing the same phosphatase. In \cite{MG2019}, authors found regions of multistationarity for cascade with 2 layers. For \eqref{eq:nlayernet}, one of the main results in this article will show that the parameter region for multistationarity in $\Sk$ is path connected (Theorem~\ref{thm:mainthm1}). This gives a lower bound on the number of disconnected regions in $\Skd$  for these cascades. We then, in particular, focus on  $(\mathcal{N}_2)$ and show that there are two connected components in the space of convex coordinates (Theorem~\ref{thm:disconnected}). As a result this article presents a concrete network where there is discrepancy in the lower and upper bound for connected regions of multistationarity in $\Skd$ obtained by the current state of the art methods. Whether the parameter region in $\Skd$ is connected for this network is an open problem (Conjecture~\ref{conj}). This network is, therefore, an interesting test case to further develop the methods for connectivity in reaction networks is $\Skd.$

	The structure of the paper is as follows: Section \ref{sec:generalrxn} recalls basic terminologies for the general cascade system of Goldbeter--Koshland loops \eqref{eq:nlayernet} and the computation of its critical polynomial. In Section \ref{sec:connected}, we study the structure of this critical polynomial an establish the connectivity in $\Sk$ for all $n\geq 2.$ In Section \ref{sec:disconnected} we focus on ($\mathcal{N}_2$) and finally, we conclude with an outlook.

	\section{\bf Cascade Networks with $n$ sites}\label{sec:generalrxn}
	
	\subsection{Cascade Network} 
	
	A general reaction network is a collection of reactions over species $X_1,\ldots, X_s$ of the form $\sum_{i=1}^sa_{ij}X_i\rightarrow \sum_{i=1}^sb_{ij}X_i$ for $j=1,\ldots, m.$ Here, $a_{ij},b_{ij}$ are non-negative integers and each reaction is weighted by reaction rate constant $\k_j\in \Rp.$ The associated stoichiometric matrix $N\in \ZZ^{s\times m}$ encodes the net production of species defined as:
	
	\[N:=(b_{ij}-a_{ij})\in \ZZ^{s\times m}\] and the reactant matrix $B\in \ZZ^{s\times m}$ is given as 
	\[B:=a_{ij}\in \ZZ^{s\times m}.\]

	In this work, we consider the family of reaction networks in \eqref{eq:nlayernet}. The species are denoted by $X_0,\ldots X_n,$ $Y_1,\ldots, Y_n,$ $U_1,\ldots,U_n,V_1,\ldots V_n, W_1,\ldots,W_n$. We note that when $n\geq 2$ for \eqref{eq:nlayernet} to exhibit multistationarity, at least two of the species $W_1,\ldots,W_n$ are identical \cite{feliu_enzyme-sharing_2011,EF2012}. In other words, $W_i\in\{H_1\ldots,H_r\}$ where $r< n$. Furthermore, we define a surjective function $\iota:\{1, \ldots, n\} \rightarrow\{1, \ldots, r\}$ with $\iota(j)=i$ if the $j$-th layer is catalyzed by $H_i$.
	For each $i=1, \ldots, r$, let $\Lambda_i=\left\{j \in\{1, \ldots, n\}: \iota(j)=i\right\}$ be the set that collects the numbers of layers that are catalyzed by $H_i$. 
	
	With this notation, we have $4n+r+1$ species given by $r$ catalysts: $H_1,\ldots,H_r$, $2n+1$ substrates: $X_0,\ldots,X_n$ and $Y_1,\ldots,Y_n$, and $2n$ intermediate species: $U_1,\ldots,U_n$ and $V_1,\ldots,V_n$.
	We denote the concentrations of the species with lowercase letters. In particular, we will often use the following notations to denote the concentration vectors of different species: $h=(h_1,\ldots,h_r)\in\Rp^r$, $x=(x_0,\ldots,x_n)\in\Rp^{n+1}$, $y=(y_1,\ldots,y_n)\in\Rp^n$, $u=(u_1,\ldots,u_n)\in\Rp^n$, and  $v=(v_1,\ldots,v_n)\in\Rp^n$.
	Assuming mass-action kinetics, the associated ODE that dictates the evolution of the concentration of the species over time is given by:

	\begin{equation}\label{eq:sys}
		\begin{aligned}
			\frac{d h_i}{d t}=&-\sum_{j \in \Lambda_i} \frac{d u_j}{d t}, \quad  i=1, \ldots, r\\
			\frac{d x_0}{d t} =& -\frac{d v_1}{d t}, \\
			\frac{d x_i}{d t}= &\; \kappa_{6i-3} v_i-\kappa_{6i-2} h_{\iota(i)} x_i+\kappa_{6i-1} u_i \\
			& +\left(\kappa_{6i+2}+\kappa_{6i+3}\right) v_{i+1}-\kappa_{6i+1} x_i y_{i+1}, \quad i=1, \ldots, n-1, \\
			\frac{d x_n}{d t} =&\; \kappa_{6n-3} v_n-\kappa_{6n-2} h_{\iota(n)} x_n+\kappa_{6n-1} u_n. \\
			\frac{d y_i}{d t}= & -\kappa_{6i-5} x_{i-1} y_i+\kappa_{6i-4} v_i+\kappa_{6i} u_i, \quad i=1, \ldots, n, \\
			\frac{d u_i}{d t}  =&\;\kappa_{6i-2} h_{\iota(i)} x_i-\left(\kappa_{6i-1}+\kappa_{6i}\right) u_i, \quad i=1, \ldots, n, \\
			\frac{d v_i}{d t} =&\; \kappa_{6i-5} x_{i-1} y_i-\left(\kappa_{6i-4}+\kappa_{6i-3}\right) v_i, \quad i=1, \ldots, n,\\
		\end{aligned}
	\end{equation}

	The right-hand side of the above equations are polynomials in $h_i,x_i,y_i,u_i,v_i$ and the coefficients are parameterised by \textit{reaction rate constants} $\k_i\in\RR_{>0}$ for $i=1,\ldots,6n.$ We denote rate-constant vector by $\kappa = (\kappa_1,\ldots,\kappa_{6n}) \in \RR^{6n}_{>0}$. This space of positive reaction rate constants is denoted by $\Sigma_\k$. For a fixed $\k \in \RR^{6n}_{>0}$, the steady state variety, $V_{\k}$ is set of all non-negative points where the right-hand side of \eqref{eq:sys} vanish.
	This network admits $n+r+1$ first linear integrals called the \textit{conservation laws}. Hence, the dynamics is contained in the subspace of dimension $3n$ determined by the following equations:
	
	\begin{equation}\label{eq:cLaws}
		\begin{aligned}
			h_i+\sum_{j \in \Lambda_i} u_j & =H_{i, t o t}, \quad i=1, \ldots, r,\\
			x_0+v_1 & =X_{0, t o t}, \\
			x_i+y_i+u_i+v_i+v_{i+1} & =X_{i, t o t}, \quad i=1, \ldots, n-1, \\
			x_n+y_n+u_n+v_n & =X_{n, t o t},
		\end{aligned}
	\end{equation}
	where $H_{i, t o t}>0$ for $i=1,\ldots,r$ and $X_{i, t o t}>0$ for $i=0,\ldots,n$. We denote the {\em total concentration constants} by ${H_{t o t}} := (H_{1, t o t},\ldots,H_{r, t o t})\in\RR^{r}_{>0}$ and ${X_{t o t}} := (X_{0, t o t},\ldots,X_{n, t o t})\in \Rp^{n+1}$. The space of all parameters given by reaction rates and the total concentration of enzymes is denoted by $\Sigma_{\k,d}:=\{(\k,H_{tot},X_{tot})~\mid~ \k\in \Sigma_{\k},H_{tot}\in \RR^r_{>0},X_{tot}\in \RR^{n+1}_{>0}\}$.
	The intersection of the level sets of equations in \eqref{eq:cLaws} with the nonnegative orthant gives the {\em stoichiometric compatibility classes}. The concentrations of species determine the corresponding stoichiometric class. In particular, we define the following map 
	\begin{equation}\label{eq:Phi}
		\begin{aligned}
			\Phi:\RR^{4n+r+1}_{>0}\to\RR^{n+r+1}_{>0}\quad \text{ such that } \quad
			(h,x,y,u,v) \mapsto ({H_{t o t}},{X_{t o t}}).
		\end{aligned}
	\end{equation}
    A reaction network is said to be \textit{conservative }if there exists a conservation relation with only positive coefficients.
    A reaction network is \textit{dissipative} if for all reaction rates and stoichiometric compatibility classes, there exists a compact set where the trajectories eventually enter at some point.
    The network \eqref{eq:nlayernet} is conservative by \eqref{eq:cLaws}, and therefore, also dissipative \cite{CC2017}.
	
	\subsection{Steady States} \label{sec:steadystate} A \textit{positive steady state} of \eqref{eq:sys} is given by $(h,x,y,u,v)\in \Rp^{4n+r+1}$ for which the right hand side of \eqref{eq:sys} vanishes. Since there are $n+r+1$ conservation laws, for fixed $\k \in \Rp^{6n}$ the equations in \eqref{eq:sys} define a variety, $V_\kappa$, of dimension $3n$ in $\Rp^{4n+r+1}$ called the \textit{steady state variety}. For a fixed reaction rate vector, a steady state $z\in V_{\kappa}$ is a \textit{relevant boundary state} if one of the coordinates of $z$ is zero and the corresponding stoichiometric compatibility class containing $z$ intersects the positive orthant.
    The steady state variety $V_\kappa$ admits a rational parametrization \cite{MG2019}.
    For $i=1,\ldots,n$, let
	\begin{equation}\label{eq:abc}
		\alpha_i=\frac{ \kappa_{6i} \kappa_{6i-2} (\kappa_{6i-4}+\kappa_{6i-3}) }{ \kappa_{6i-3} \kappa_{6i-5} (\kappa_{6i-1}+\kappa_{6i}) },\qquad
		\beta_i=\frac{\kappa_{6i-2}}{\kappa_{6i-1}+\kappa_{6i}},\qquad
		\gamma_i=\frac{ \kappa_{6i} \kappa_{6i-2} }{ \kappa_{6i-3} (\kappa_{6i-1}+\kappa_{6i}) }.
	\end{equation}
	For fixed $\kappa\in \Rp^{6n}$ and $(h,x)\in \Rp^{n+r+1}$, \eqref{eq:sys} has a unique solution given by:
	\begin{equation}\label{eq:ssp}
		\begin{aligned}
			y_i =\alpha_i h_{\iota(i)} x_{i-1}^{-1} x_i, \qquad 
			u_i =\beta_i h_{\iota(i)} x_i, \qquad
			v_i =\gamma_i h_{\iota(i)} x_i. 
		\end{aligned}
	\end{equation}
	Consequently, we get the rational parametrization of the set of positive states in terms of $h$ and $x$ denoted by $\varphi_\kappa$. In particular,
	\begin{equation}\label{eq:varphi}
		\begin{aligned}
			\varphi_{\kappa}:\RR^{n+r+1}_{>0}&\to\RR^{4n+r+1}_{>0}.
		\end{aligned}
	\end{equation}
	We further define the following map: 
	\begin{equation}\label{eq:psi}
		\begin{aligned}
			\psi_{\kappa}:\RR^{4n+r+1}_{>0}\to\RR^{4n+r+1}_{>0},
		\end{aligned}
	\end{equation}
	whose first $n+r+1$ entries are given by the left-hand side of $\eqref{eq:cLaws}$ and the remaining $3n$ entries are given by $\frac{d y_1}{d t},\ldots,\frac{d y_n}{d t},\frac{d u_1}{d t},\ldots,\frac{d u_n}{d t},\frac{d v_1}{d t},\ldots,\frac{d v_n}{d t}$.

    
	A rate-constant vector $\k\in\Rp^{6n}$ {\em enables multistationarity} if there exists a stoichiometric compatibility class determined by $H_{tot}$ and $X_{tot}$ whose intersection with the steady state variety gives more than one positive steady state. Identifying parameter regions that enable multistationarity and studying their shapes is an active area of research. In this article, our main result will show that all rate-constant vectors in $\Rp^{6n}$ enable multistationarity in \eqref{eq:nlayernet} for $n\geq 2$ and $r<n$. Consequently, this parameter region is also connected.
	
	\subsection{Criteria for Multistationarity} \label{sec:multicriteria} Given a polynomial map $g:\RR^m \to \RR^m$, we denote by $J_g(z)$ the Jacobian of the map where $z \in \RR^m.$ Using this notation, we define the following function associated with our network \eqref{eq:nlayernet}
	\begin{equation}\label{eq:crit}
		q_{\kappa}\hx:=(-1)^{3 n} \operatorname{det}\left( J_{\psi_\kappa}(\varphi_\kappa \hx)\right),
	\end{equation}
	where $\k \in \Rp^{6n}$, and $\varphi_{\kappa}$ and $\psi_{\kappa}$ are the functions defined in \eqref{eq:varphi} and \eqref{eq:psi} respectively. In particular, we consider the Jacobian of the system evaluated at the positive steady states. 
   Since the network \eqref{eq:nlayernet} is dissipative and has no relevant boundary steady state by \cite[Corollary 2]{CatalystInterm}, the following proposition is a direct consequence of \cite[Theorem 1]{CC2017}.
	
	\begin{proposition}\label{prop:cp<0}
		For fixed $\kappa^* \in \RR^{6n}_{>0}$, if there exists
		$\hxs\in \Rp^{n+r+1}$ such that $q_{\kappa^*} \hxs<0$,
		then $\kappa^*$ enables multistationarity.
	\end{proposition}
	
	Based on Proposition \ref{prop:cp<0}, for a fixed $\k$ deciding multistationarity is equivalent to finding $\hx \in \RR^{n+r+1}_{>0}$ such that the sign of $q_{\k}\hx$ is negative.
	In the next proposition, we simplify the sign computation of this polynomial by reducing it to the determinant of the matrix of size $n+r+1$ as opposed to that of size $4n+r+1$.
	For any $z \in \RR$, $\text{sign}(z)$ is $-1,0,$ or $1$ if $z<0,z=0,$ or $z>0$ respectively.
	
	\begin{proposition}\label{prop:cp=detJ}
		Let $\Phi$, $\varphi_{\kappa}$, and $\psi_\k$ be the functions defined in \eqref{eq:Phi},\eqref{eq:varphi}, and \eqref{eq:psi}, respectively.
		For any $\kappa \in \RR^{6n}_{>0}$ 
		and for any 
		$\hx\in \RR^{n+r+1}_{>0}$, we have
		\[\sign(q_{\kappa}\hx)=\sign(\det \left( J_{\Phi \circ \varphi_\kappa}\hx \right)).\]
	\end{proposition}
	
	\begin{proof}
		We first consider the map $\psi_\kappa \circ \varphi_\kappa$. The first $n+r+1$ components of $\psi_\kappa \circ \varphi_\kappa$ and $\Phi \circ \varphi_\kappa$ are the same and the remaining $3n$ components are 0. Hence,
		
		\begin{equation}\label{eq:cp1}
			J_{\psi_\kappa \circ \varphi_\kappa}\hx=
			\begin{pmatrix}
				J_{\Phi \circ \varphi_\kappa}\hx \\
				0_{3n \times (n+r+1)}
			\end{pmatrix}
		\end{equation}
		Furthermore, by the chain rule, we have
		\begin{equation}\label{eq:cp2}
			J_{\psi_\kappa \circ \varphi_\kappa}\hx = J_{\psi_\kappa}\left(\varphi_\kappa\hx\right) J_{\varphi_\kappa}\hx.
		\end{equation}
		Let $\Id_{n+r+1}$ the identity matrix of size $n+r+1$ and $M$ be a matrix of size $3n \times (n+r+1)$.
		We rewrite $J_{\varphi_\kappa}$ as
		\begin{equation}\label{eq:cp4}
			J_{\varphi_\kappa}\hx=
			\begin{pmatrix}
				\Id_{n+r+1} \\
				M
			\end{pmatrix},
		\end{equation}
		and we rewrite $J_{\psi_\kappa}\left(\varphi_\kappa\hx\right)$ in block form as
		\begin{equation}\label{eq:cp3}
			J_{\psi_\kappa}\left(\varphi_\kappa\hx\right)=
			\begin{pmatrix}
				A&B \\
				C&D
			\end{pmatrix},
		\end{equation}
		where  $A \in \RR^{(n+r+1)\times (n+r+1)}, \;B \in \RR^{(n+r+1)\times 3n}, \;C \in \RR^{3n \times (n+r+1)},$ and $D \in \RR^{3n\times 3n}$.

		From the above equalities, we get
		\begin{equation}\label{eq:cp5}
			\begin{pmatrix}
				J_{\Phi \circ \varphi_\kappa}\hx \\
				0_{3n \times (n+r+1)}
			\end{pmatrix}
			=
			\begin{pmatrix}
				A&B \\
				C&D
			\end{pmatrix}
			\begin{pmatrix}
				\Id_{n+r+1} \\
				M
			\end{pmatrix}.
		\end{equation}
		Consequently,
		\begin{equation}\label{eq:cp6}
			J_{\Phi \circ \varphi_\kappa}\hx
			=A+B M \qquad  0=C +DM.
		\end{equation}
		By \eqref{eq:cp3}, \eqref{eq:cp5}, and \eqref{eq:cp6}, we have
		\begin{equation*}
			\begin{aligned}
				\det \left( J_{\psi_\kappa}\left(\varphi_\kappa\hx\right) \right)
				=&
				\begin{vmatrix}
					A&B \\
					C&D
				\end{vmatrix}
				=
				\begin{vmatrix}
					A+BM&B \\
					C+DM&D
				\end{vmatrix}
				=
				\begin{vmatrix}
					J_{\Phi \circ \varphi_\kappa}\hx&B \\
					0_{3n\times (n+r+1)}&D
				\end{vmatrix}
				=
				\det \left( J_{\Phi \circ \varphi_\kappa}\hx \right) 
				\det (D).
			\end{aligned}
		\end{equation*}
		Note that the matrix $D$ is obtained by considering the Jacobian of the last $3n$ entries of the map $\psi_\k$ with respect to the variables $y,u,v$. To finish the proof, it is enough to show that the sign of $\det (D)$ is $(-1)^{3n}$. To that end, consider the following digraph rooted at the vertex $\star$:
		
		
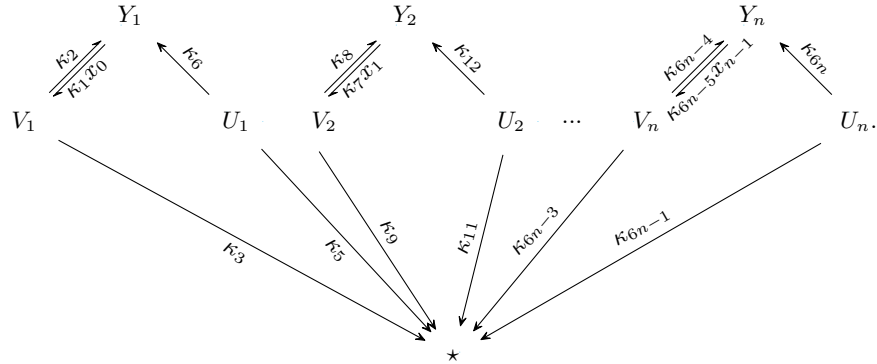
\vspace{0.5cm}
		\begin{scriptsize}
			\begin{tabular}{ccccccccc}
				& & & & & &
				\schemestart
				$V_1$ 
				\arrow{<=>[$\kappa_{2}$][$\kappa_{1} x_0$]}[45]
				$Y_1$
				\arrow{<-[$\kappa_{6}$]}[315]
				$U_1$
				\arrow{0}[,0.1]
				\qquad 
				$V_2$ 
				\arrow{<=>[$\kappa_{8}$][$\kappa_{7} x_1$]}[45]
				$Y_2$
				\arrow{<-[$\kappa_{12}$]}[315]
				$U_2$
				\arrow{0}[,0.1]
				\quad...\quad$\;\;$
				\arrow{0}[,0.1]
				$V_n$ 
				\arrow{<=>[$\kappa_{6n-4}$][$\kappa_{6n-5} x_{n-1}$]}[45]
				$Y_n$
				\arrow{<-[$\kappa_{6n}$]}[315]
				$U_n$.
				\arrow{->[$\kappa_{6n-1}$]}[-150,3.6]
				\arrow(@c1--@c11){->[][$\kappa_{3}$]}
				\arrow(@c3--@c11){->[][$\kappa_{5}$]}
				\arrow(@c4--@c11){->[$\kappa_{9}$]}
				\arrow(@c6--@c11){->[$\kappa_{11}$]}
				\arrow(@c8--@c11){->[$\kappa_{6n-3}$]}
				$\star$
				\schemestop
			\end{tabular}
		\end{scriptsize}
		
		\medskip

		\noindent 
		It is straightforward to check that the matrix $D$ can also be obtained by deleting the row and column corresponding to the vertex $\star$ of the Laplacian of the digraph.
		Hence, as a consequence of the Matrix-Tree theorem the determinant of $D$ is $(-1)^{3n}$ times the sum of the labels of all spanning trees rooted at $\star$ (cf. \cite[Proposition 1]{MS2018}).
		Since there is at least one spanning tree rooted at $\star$, and the labels of all the edges are positive, the sign of $\det (D)$ is $(-1)^{3n}$. This concludes the proof.
	\end{proof}
	
	By Proposition~\ref{prop:cp<0} and Proposition~\ref{prop:cp=detJ} to determine whether a reaction rate vector enables multistationarity, it is enough to consider the sign of the following function
	\begin{align}p_{\k}\hx:=\det \left( J_{\Phi \circ \varphi_\kappa}\hx \right).\end{align}
	We will discuss the structure of the matrix $ J_{\Phi \circ \varphi_\kappa}\hx$ in Section~\ref{sec:str} and then use it to study the sign of the determinant. 
	To that end, in the next subsection, we highlight some basic results from combinatorial geometry that we will use to study these signs. 
	
	\subsection{Newton Polytopes}\label{sec:NP} Consider a multivariate function $f:\Rp^m\to\RR$ given by 
	\begin{align}\label{eq:f}f(z):=\sum_{\mu\in \ZZ^m}c_\mu z^\mu. \end{align}
	The Newton polytope of $f$, denoted by $\np(f)$, is given by the convex hull of the exponents of monomials with non-zero coefficients. In other words, $\np(f):=\conv(\{\mu\in\ZZ^m~\mid ~c_\mu \neq 0\}).$
	We call $\mu$ a {\em positive exponent} (respectively, {\em negative exponent}) of $\np(f)$, if the coefficient $c_{\mu}$ of $x^\mu$ is  positive (resp. negative). Let $\Omega$ be the set of all the exponent vectors of $f$. If there exists a hyperplane $L_{\lambda,\ell}(z):=z\cdot\lambda+\ell$ such that $L_{\lambda,\ell}(\mu)\geq 0$ for some $\mu \in \Omega$ and $L_{\lambda,\ell}(\nu)<0$ for all $\nu \in \Omega \setminus \mu$, then $\mu$ is a vertex of $\np(f)$ and $L_{\lambda,\ell}(z)$ is a \textit{separating hyperplane} of $\mu$. We will repeatedly use the following well-known lemma that guarantees the polynomial can reach a positive (resp. negative) value on the positive orthant if its Newton polytope has a positive (resp. negative) vertex.
	
	\begin{lemma}\label{lem:negvertex}
		Consider $f(z)$ as in \eqref{eq:f}.
		If $\np(f)$ has a positive vertex (resp. negative vertex), then there exists $z^* \in \RR^m_{>0}$ such that $f(z^*)>0$ (resp. $f(z^*)<0$).
	\end{lemma}
	
	We illustrate this with the following example.
	
	\begin{example}
		Consider a polynomial $f(z)=4z_1^2 z_2^2 - z_1^2 + 2z_2^2 + 1$.
		The exponent vectors of $f(z)$ are given by $\{(2,2),(2,0),(0,2),(0,0)\}$.
		All the exponents of $\np(f)$ are the vertices. Note the polynomial has both negative and positive vertex. In the positive orthant, $f(z)$ attains positive value at $(1,1)$ and it attains negative value at at $(3,\frac{1}{3})$.
	\end{example}
	
	We remark that the explicit points at which the polynomial attains such values can be found using the outer normal cones of the Newton polytopes. For detailed exposition on convex polytopes, the readers are referred to \cite{book-grunbaum}. In Section~\ref{sec:connected}, we will first consider the structure of the Jacobian matrix $J_{\Phi \circ \varphi_\kappa}\hx$ and its determinant and then use it to determine a negative vertex in the newton polytope which will then allow us to employ Lemma~\ref{lem:negvertex}.

	\section{\bf Connectedness in $\Sigma_k$}\label{sec:connected}
	
	In this section, for arbitrarily fixed $n\geq 2$, our main aim is to prove that the region of parameters in $\Sigma_\k$ that enable multistationarity for the network \eqref{eq:nlayernet} is path connected. In fact, we will show that it is all of $\Sigma_\k$. In particular, we aim to prove the following result:
	
	\begin{theorem}\label{thm:mainthm1}
		For $n\geq 2$, $r<n$, and for every $\k \in \Rp^{6n}$ there exists $\hxs \in \Rp^{n+r+1}$ such that $p_{\k}\hxs<0.$ In other words, all $\k \in \Rp^{6n}$ enables multistationarity in \eqref{eq:sys}.
	\end{theorem}
	
	To prove the result, we further treat $h=(h_1,\ldots, h_r)\in \Rp^r$ as parameters and rewrite $p_{\k}(h,x)$ as $p_{\k,h}(x)$. With these notations in place, we will prove the following result:
	
	\begin{theorem}\label{thm:mainthm2}
		For $n\geq 2$, $r<n$, and for every $\k \in \Rp^{6n}$ and $h \in \Rp^r$ there exists $x^*\in \Rp^{n+1}$ such that  $p_{\k,h}(x^*)<0$. 
	\end{theorem}
	
	Note that Theorem~\ref{thm:mainthm1} is a direct consequence of Theorem~\ref{thm:mainthm2}. To prove this result, it is enough to show that the corresponding Newton polytope of $p_{\k,h}$ has a negative vertex, and then the result follows by Lemma~\ref{lem:negvertex}.
	
	\subsection{Structure of $p_{\k,h}(x)$}\label{sec:str} To study the Newton polytope of $p_{\k,h}$, we need to study $\det(J_{\Phi\circ\varphi_\k})$. Recall that $\Phi\circ\varphi_\k$ is determined by the left-hand side of \eqref{eq:cLaws} after substituting the rational parametrization given in \eqref{eq:ssp}. The Jacobian of this map is a matrix of size $n+r+1$. We write this matrix as below:
	
	\begin{equation}\label{eq:Jacobian}
		J_{\Phi\circ\varphi_\k}=
		\begin{pmatrix}
			J_1'&J_2' \\
			J_3'&J_4'
		\end{pmatrix},
	\end{equation}
	where $J_1'$ and $J_4'$ are matrices of size $r$ and $n+1$ respectively. The entries of submatrices are described below. The rows and columns of $J_4'$ are labeled from $0,\ldots,n.$ Since we want to work with $p_{\k,h}(x)$ with only $x$ as variable, we set the following notation:
	
	\begin{equation}\label{eq:abc1}
		a_i:=	\alpha_ih_{\iota(i)},\qquad
		b_i:=	\beta_ih_{\iota(i)},\qquad
		c_i:=	\gamma_ih_{\iota(i)}.
	\end{equation}
	
	The matrix $J_1'$ is a diagonal matrix of size $r$ with the following entries:
	
	\begin{equation}
		(J_1')_{ii}= 1+\sum_{k\in \Lambda_i}\beta_kx_k.
	\end{equation}
	
	The matrix $J_4'$ is:
	
	\begin{equation}
		(J_4')_{ij}= \begin{cases} 1  &i=j=0\\
			1+a_i	x_{i-1}^{-1}+ b_i+c_i  \qquad &i=j \text{ and } j\in \{1,\ldots, n\}\\
			c_j &i+1=j \text{ and } j\in \{1,\ldots, n\}\\
			-a_ix_ix_{i-1}^{-2} &i=j+1 \text{ and } j\in \{0,\ldots, n-1\} \\
			0 &else
		\end{cases}.
	\end{equation}

	The matrices $J_2'$ and $J_3'$ have sizes $r \times (n+1)$ and $(n+1) \times r$ respectively. We label the rows of $J_2'$ from $1,\ldots,r$ and columns from $0,\ldots,n$. Similarly, we label the rows of $J_3'$ from $0,\ldots,n$ and columns from $1,\ldots,r$. Going forward, we define the following notation:
	
	\begin{equation}
		\delta_{i,j}=\begin{cases}
			1 \qquad &i \in \Lambda_j\\
			0 \qquad &else
		\end{cases}.
	\end{equation}
	
	Additionally, we set $\delta_{-1,j}=\delta_{0,j}=0=\delta_{i,0}$ for all $i,j$. The entries of $J_2'$ and $J_3'$ are now given by:
	
	\begin{equation}
		(J_2')_{ij}=	\delta_{j,i}b_j, \end{equation}
	\begin{equation}	(J_3')_{ij}= 
		\delta_{i,j} \big(\alpha_i x_i x_{i-1}^{-1}+(\beta_i+\gamma_i)x_i\big)+\delta_{i+1,j}\big(\gamma_{i+1}x_{i+1}\big) .
	\end{equation}

	Next, we carry out some matrix manipulation to compute the determinant. For $j\in \{1,\ldots,n\}$, subtract $h_{\iota(j)}^{-1} x_j$ times the $(r+1+j)$-th column of $J_{\Phi\circ \varphi_\k}$ from the $\iota(j)$-th column of $J_{\Phi\circ \varphi_\k}$. This gives the following matrix:
	
	\begin{equation}
		\label{eq:Jacobian1}
		J=
		\begin{pmatrix}
			\Id_{r}&J_2' \\
			J_3&J_4'
		\end{pmatrix},
	\end{equation}
	where the entries of $J_3$ are as follows: 
	
	\begin{equation*}
		(J_3)_{ij}=-\delta_{i,j}x_ih^{-1}_{\iota(i)}+\delta_{i-1,j}a_ix_ix_{i-1}^{-1}h^{-1}_{\iota(i-1)}.
	\end{equation*}
	Note that 
	
	\begin{equation}
		\label{eq:Jacobian2}
		\det J_{\Phi\circ\varphi_\k} = \det J=
		\begin{vmatrix}
			\Id_{r}&J_2' \\
			J_3&J_4'
		\end{vmatrix}=\begin{vmatrix}
			\Id_{r}&J_2' \\
			0&J_4
		\end{vmatrix}=\det J_4,
	\end{equation}
	where $J_4:=J_4'-J_3J_2'$ is a matrix of size $n+1$.
	It is direct to check entries of $J_3J_2'$ is given by:
	
	\begin{equation*}
		(J_3J_2')_{ij}=\sum_{k=1}^r(-\delta_{i,k}\delta_{j,k}\beta_jx_i+\delta_{i-1,k}\delta_{j,k}\beta_ja_ix_ix_{i-1}^{-1})
	\end{equation*}
	and hence, entries of $J_4$ are given by:
	{\small
		\begin{equation}\label{eq:J4}
			(J_4)_{ij}=\begin{cases} 1  &i=j=0\\
				1+a_i	x_{i-1}^{-1}+ b_i+c_i+\beta_ix_i-\sum_{k=1}^{r} \delta_{i-1,k}\delta_{i,k}\beta_ia_ix_ix^{-1}_{i-1}  \qquad &i=j \text{ and } j\in \{1,\ldots, n\}\\
				c_j +\sum_{k=1}^r(\delta_{i,k}\delta_{i+1,k}\beta_{i+1}x_i-\delta_{i-1,k}\delta_{i+1,k}\beta_{i+1}a_ix_ix_{i-1}^{-1}) &i+1=j \text{ and } j\in \{1,\ldots, n\}\\
				-a_ix_ix_{i-1}^{-2}+\sum_{k=1}^r(\delta_{i-1,k}\delta_{i,k}\beta_{i-1}x_i-\delta_{i-1,k}\beta_{i-1}a_ix_ix_{i-1}^{-1}) &i=j+1 \text{ and } j\in \{0,\ldots, n-1\} \\
				\sum_{k=1}^r(\delta_{i,k}\delta_{j,k}\beta_{j}x_i-\delta_{i-1,k}\delta_{j,k}\beta_{j}a_ix_ix_{i-1}^{-1}) & else
			\end{cases}.
		\end{equation}
	}
	
	We want to compute the determinant of $J_4$ and show that for all $\k,h$ there exists $x\in \Rp^{n+1}$ such that the determinant is negative. Note that $p_{\k,h}(x)=\det (J_4)$. In the remaining section, we will treat $p_{\k,h}$ and each of the entry in $J_4$ as an element in $R:=\RR(\k,h)[x,x^{-1}]$ where $x^{-1}:=(x_0^{-1},\ldots x_n^{-1})$. For brevity, from now on we will simply write $p(x)$ in place of $p_{\k,h}(x).$
	
	\begin{remark}\label{rem:obs}  We make following observations on the entries of $J_4$ that we will recall repeatedly later.
		\begin{itemize}
			\item[(i)] From the above description, note that row 0 has only constant terms of the ring $R$. 
			\item[(ii)] For $i=1,\ldots n$, each variable $x_i$ only appears in row $i$ and row $i+1$. Equivalently, in the $i-$th row only variables $x_i$ and $x_{i-1}$ appear.
			\item[(iii)] For $i=1,\ldots n-1$, each variable $x_i$ appears in row $i$ with degree 0 or 1 and in row $i+1$ with degree $-2, -1$ or 0.
			\item[(iv)] For $i=0,\ldots,n-1,$ variable $x_i$ appears with degree $-2$ only once in row $i+1$ and column~$i$.
			\item[(v)] Variable $x_0$ only appears in row 1 with degrees $-2, 1$, or 0.
			\item[(vi)] Variable $x_n$ only appears in row $n$ with degrees 1 or 0.
			
		\end{itemize}
	\end{remark}
	
	\begin{remark}\label{rem:matrixreduced}
		Note that the first row and the first column of the matrix $J_4$ have only two entries. To get the determinant, we can expand along the first column to obtain the following:
		\begin{equation}\label{eqn:determinantsums}
			\det(J_4)= \det(A_1)+ c_1(-a_1x_1x_0^{-2})\det(A)
		\end{equation}
		where $A_1$ is an $n\times n$ matrix obtained by removing the first row and column of $J_4$ and $A$ is an $(n-1)\times (n-1)$ obtained by removing the first two rows and the first two columns of $J_4$. Furthermore, using this expansion and by Remark~\ref{rem:obs}(v), it is easy to check that $x_0$ appears with degree $-2$ only in terms obtained by the expansion of the second summand of the above equation.
	\end{remark}
	
	Let $p(x)=\det(J_4):=\sum_{\mu}\theta_\mu x^{\mu} \in R$. To show that $p$ attains negative values over the positive orthant, by Lemma~\ref{lem:negvertex} it is enough to show that the Newton polytope of $p$ has a negative vertex. Let $\mu \in \ZZ^{n+1}$ be a vector such that $\theta_\mu \neq 0$, then the following proposition accounts for all such possible vectors $\mu$.
	
	\begin{proposition}\label{prop:mu}
		Let $p(x)=\det(J_4)=\sum_{\mu}\theta_\mu x^{\mu} \in R$ and $\mu:=(\mu_0,\ldots, \mu_{n}) \in \ZZ^{n+1}$ be such that $\theta_\mu \neq 0$. Then, the following holds:
		\begin{itemize}
			\item[(i)] For all $i=1,\ldots,n-1$, $\mu_i\in \{-2,-1,0,1\}$, $\mu_{0} \in \{-2,-1,0\}$, and $\mu_n\in \{0,1\}$.
			\item[(ii)] If $\mu_0=-2$, then $\mu_1\in \{0,1\}.$
			\item[(iii)] If $\mu_0=-2$ and $\mu_1=0$, then $\mu_2\neq -2.$
			\item[(iv)] For $i\in \{0,\ldots,n-1\}$, if $\mu_i\in\{-2,-1\}$, then $\mu_{i+1}\neq -2.$
		\end{itemize}
	\end{proposition}

	\begin{proof}
		Part (i) follows directly from Remark~\ref{rem:obs}. From Remark~\ref{rem:obs}(ii), (iii), and (iv), we can deduce $\mu_i\in \{-2,-1,0,1\}$ for all $i=1,\ldots,n-1$. Remark~\ref{rem:obs}(v) and (vi) ensure that $\mu_{0} \in \{-2,-1,0\}$, and $\mu_n\in \{0,1\}$.
		
		Part (ii) follows from Remark~\ref{rem:matrixreduced} and \ref{rem:obs}(iii).
		
		For part (iii), note that $x_0$ appears with degree $-2$ from the second summand of Equation~\eqref{eqn:determinantsums}. Moreover, in that summand, we obtain $\mu_1=0$ from the entries of $A$ where $x_1$ appears with degree $-1$. From \eqref{eq:J4}, entries of $A$ where $x_1$ appears with degree $-1$ have monomials of the form $x_1^{-1}$ or $x_2x_1^{-1}$. In the first case, the monomial $x_1^{-1}$ appears in the same column as the monomial $x_3x_2^{-2}$ and hence, we can never get $\mu_2=-2.$ In the second case, since the degree of $x_2$ in any entry is no less than $-2$, on expanding by Remark~\ref{rem:obs}(iii) we get $\mu_2 \geq -1$ when $\mu_1=0$ and $\mu_0=-2.$
		
		To see part (iv) holds, note that when $i=n-1$, $\mu_{i+1}=\mu_n \geq 0$ by Remark~\ref{rem:obs}. For the remaining part, let $i\in\{0,\ldots,n-2\}.$ 
		
		First, let $\mu_i=-2$, from \eqref{eq:J4}, it is easy to check that this only appears once in a single entry of $J_4$ given by $a_{i+1}x_{i+1}x_{i}^{-2}$. Since all the remaining terms have degree of $x_{i+1}$ greater than or equal to $-2$, on expansion, by Remark~\ref{rem:obs}(iii) we get $\mu_{i+1}\geq {-1}$ whenever $\mu_i=-2.$
		
		Now, let $\mu_i=-1$. From Remark~\ref{rem:obs}(iii), this can be obtained in two ways that are explained next. (a) In the first case, we obtain $\mu_i=-1$ by multiplying the unique monomial, $x_{i+1}x_{i}^{-2}$, containing $x_i$ with degree $-2$ in row $i+1$ of $J_4$ with monomials containing $x_i$ with degree 1 in row $i$. By Remark~\ref{rem:obs}, row $i+2$ has monomials such that $x_{i+1}$ has degree $\geq -2.$ On multiplication with $x_{i+1}x_{i}^{-2}$ this ensures $\mu_{i+1}\geq {-1}$.
		
		(b) In the second case, we consider the entries of $J_4$ in row $i+1$ where $x_i$ appears with degree $-1$. These contain monomials are of the form $x_i^{-1}$ or $x_{i+1}x_i^{-1}$. In the first case, the term $x_i^{-1}$ appears in the same column as the term containing $x_{i+1}^{-2}$ and hence, we can never get $\mu_{i+1}=-2.$ In the second case, since the degree of $x_{i+1}$ in any entry is $\geq -2$, on expanding by Remark~\ref{rem:obs} we get $\mu_{i+1} \geq -1$.
		\end{proof}
	
	\subsection{Negative vertex of $\np(p)$}\label{sec:nv}
	To prove Theorem~\ref{thm:mainthm2}, we need to show that $p(x)$ attains negative values over the positive orthant for all $\k\in \RR^{6n}_{>0}$ and $h\in \RR^r_{>0}$. To that end, in the Section~\ref{sec:str} we looked at the structure of  $p_{\k,h}(x).$ Proposition~\ref{prop:mu} gives conditions on exponent vectors of the monomials of $p(x).$ We will now use this to find a monomial whose coefficient is negative and its exponent vector is a vertex of $\np(p).$ Then, Theorem~\ref{thm:mainthm2} will follow from Lemma~\ref{lem:negvertex}.
	
	\begin{proposition}\label{prop:vtx}
		Let $n \geq 2$ and 
		\begin{equation*}
			\omega=(-2,0,\underbrace{-1,\ldots,-1}_{k},\underbrace{1,\ldots,1}_{n-k-1})
		\end{equation*}
		for some $k \in \{0,\ldots,n-2\}$.
		If $\theta_{\omega} \neq 0$, then $\omega$ is a vertex of $N(p)$.
	\end{proposition}
	\begin{proof}
		Consider the hyperplane $L_{\lambda,\ell}(z):=z\cdot\lambda+\ell$, where
		\begin{equation*}
			\lambda=(-(k+2),-(k+1),\ldots,-1,\underbrace{1,\ldots,1}_{n-k-1})^T.
		\end{equation*} 
		Note that if $\omega \cdot \lambda > \mu \cdot \lambda$ for all $\mu \in \Omega \setminus \omega$ then we can set $\ell = -\omega \cdot \lambda$ and in this case, it is easy to check that $L_{\lambda,\ell}$ is a separating hyperplane of $\omega$. 
		To show $\omega$ is a vertex of $\np(p)$, it is now enough to show that 
		$\omega \cdot \lambda > \mu \cdot \lambda$ for all $\mu \in \Omega \setminus \omega$.
		
		Let $\mu:=(\mu_0,\mu_1,\ldots, \mu_n)$ be an arbitrary exponent vector in $\Omega.$ Then, by Proposition~\ref{prop:mu}, we have that 
		\begin{align*}(\mu_0,\mu_1)\in \{(-2,0),(-2,1),(-1,-1),(-1,0),(-1,1),(0,-2),(0,-1),(0,0),(0,1)\}.
		\end{align*}
		It is straightforward to check that  
		\begin{align}\label{eq:l1}
			\l_0\omega_0+\l_1\omega_1\geq \l_0\mu_0+\l_1\mu_1
		\end{align}
		and the equality holds iff $(\mu_0,\mu_1)=(\omega_0,\omega_1)$.
		
		Similarly, by Proposition~\ref{prop:mu}(i), the maximum value of an entry in any exponent vector is 1 and hence, 
		\begin{equation}\label{eq:l2}
			\sum_{i=k+2}^n \l_i \omega_i \geq \sum_{i=k+2}^n \l_i \mu_i,
		\end{equation}
		and equality holds if and only if $\mu_i = \omega_i$ for all $i \in \{k+2,\ldots,n\}$.
		
		Furthermore, let $J \subset \{2,\ldots,k+1\}$, such that $\mu_j \ne -2$ for $j \in J$. Let $I_1 \subset \{3,\ldots,k+1\}$ be the set such that $\mu_i = -2$ for all $i \in I_1$ and $\mu_i \ne -2$ for any $i \in \{3,\ldots,k+1\} \setminus I_1$.
		By $I_{2}$ denote the set obtained by subtracting 1 from each element of $I_1$. By Proposition \ref{prop:mu} (iv), we have $I_{1} \cap I_2 = \emptyset$. Using this, we have
		\begin{equation}\label{eq:l3}
			\sum_{j \in J} \lambda_j \omega_j=\sum_{j \in J} -(k+2-j)(-1) \geq \sum_{j \in J} -(k+2-j) \mu_j = 	\sum_{j \in J} \lambda_j \mu_j 
		\end{equation}
		and the equality holds if and only if $\mu_j = \omega_j$ for all $j \in J$.  Moreover, for $i \in I_1$, by Proposition~\ref{prop:mu} (iv), we have $\mu_{i-1} = 0 \text{ or } 1$. In either of these cases, simple computation gives
		\begin{equation}\label{eq:l4}
			\lambda_{i-1} \omega_{i-1} + \lambda_i \omega_i >	\lambda_{i-1} \mu_{i-1} + \lambda_i \mu_i .
		\end{equation}

		With the above inequalities in place, note that if $k=0$, then by \eqref{eq:l1} and \eqref{eq:l2}, we have 
		\begin{equation*}
			\omega \cdot \l=\sum_{i=0}^1 \lambda_i \omega_i+\sum_{i=2}^n \l_i \omega_i \geq \sum_{i=0}^1 \lambda_i \mu_i+\sum_{i=2}^n \l_i \mu_i = \mu \cdot \lambda.
		\end{equation*}
		The equality holds if and only if $\mu = \omega$. 
		Next, we consider the cases when $k\neq 0$. We will split this into two cases (a) when $\mu_2\neq-2$ and (b) when  $\mu_2 = -2$.

		{\bf Case (a).} Assume that $\mu_2\neq-2$. By \eqref{eq:l1}, \eqref{eq:l2}, \eqref{eq:l3} and \eqref{eq:l4}, we get the desired inequality as below
		\begin{align*}
			\omega\cdot \l &= \sum_{i=0}^1 \lambda_i \omega_i + \sum_{i \in I_1} (\lambda_{i-1} \omega_{i-1} + \lambda_i \omega_i) + \sum_{i \in \{2,\ldots,k+1\} \setminus (I_{1} \cup I_2)} \lambda_i \omega_i + \sum_{i=k+2}^n \omega_i\\
			&\geq  \sum_{i=0}^1 \lambda_i \mu_i + \sum_{i \in I_1} (\lambda_{i-1} \mu_{i-1} + \lambda_i \mu_i) + \sum_{i \in \{2,\ldots,k+1\} \setminus (I_{1} \cup I_2)} \lambda_i \mu_i + \sum_{i=k+2}^n \mu_i \\
			&= \mu \cdot \l,
		\end{align*}
		where the equality holds if and only if $\mu=\omega.$ 
		
		{\bf Case (b).} Assume that $\mu_2 = -2$.
		By Proposition \ref{prop:mu}, we get 
		\begin{align*}
			(\mu_0,\mu_1,\mu_2)	\in \{(-2,1,-2),(-1,0,-2),(-1,1,-2),(0,0,-2),(0,1,-2)\}.
		\end{align*}
		Note that, in this case, the following holds
		\begin{equation}\label{eq:l5}
			\lambda_0 \omega_0+\lambda_1 \omega_1 +\lambda_2 \omega_2 >	\lambda_0 \mu_0+\lambda_1 \mu_1+\lambda_2 \mu_2 
		\end{equation}
		By Proposition \ref{prop:mu} (iv), $\mu_3 \ne -2$.
		Therefore, by \eqref{eq:l2},\eqref{eq:l3},\eqref{eq:l4}, and \eqref{eq:l5} we get 
		\begin{align*}
			\omega\cdot \l &= \sum_{i=0}^2 \lambda_i \omega_i + \sum_{i \in I_1} (\lambda_{i-1} \omega_{i-1} + \lambda_i \omega_i) + \sum_{i \in \{3,\ldots,k+1\} \setminus (I_{1} \cup I_2)} \lambda_i \omega_i + \sum_{i=k+2}^n \omega_i\\
			&>  \sum_{i=0}^2 \lambda_i \mu_i + \sum_{i \in I_1} (\lambda_{i-1} \mu_{i-1} + \lambda_i \mu_i) + \sum_{i \in \{3,\ldots,k+1\} \setminus (I_{1} \cup I_2)} \lambda_i \mu_i + \sum_{i=k+2}^n \mu_i \\
			&= \mu \cdot \l.
		\end{align*}
		This concludes the proof.
	\end{proof}

	Denote by $\hat{i}$ the largest index of the layer that shares phosphatase with other layers. When no phosphatase is shared, then we do not have multistationarity \cite{EF2012}. We therefore, assume that $\hat{i} \geq 2$.

	\begin{proposition}\label{lemma:nvtx}
		If $\hat{i} \ne 0$, 
		then the vector 
		\begin{equation}\label{eq:omega}
			\omega=(-2,0,\underbrace{-1,\ldots,-1}_{\hat{i}-2},\underbrace{1,\ldots,1}_{n-\hat{i}+1})^T
		\end{equation}
		is a negative vertex of $N(f)$.
	\end{proposition}  
	
	\begin{proof}
		We need to show that the coefficient, $\theta_{\omega}$, of the monomial whose exponent vector is given by $\omega$ is negative. Since $\omega_0=-2$, by Remark~\ref{rem:obs}, the monomial only appears in the expansion of the second summand in Equation~\eqref{eqn:determinantsums}. 
		
		Note that row $i$ has monomials of the form $x_i,x_ix_{i-1}^{-1},x_ix_{i-1}^{-2},x_{i-1}^{-1}$ and constants. We claim that to obtain $\omega$, there exists some $i \in \{2,\ldots,n\}$, such that the monomial $x_ix_{i-1}^{-1}$ is selected in row $i$ for multiplication. Indeed, if this is not the case, then only $x_i,x_ix_{i-1}^{-2},x_{i-1}^{-1}$ and constants can be selected.
		By Remark \ref{rem:obs} (ii) and \eqref{eq:J4}, for $i \in \{2,\ldots,\hat{i}\}$, $x_{i-1}^{-1}$ is selected in row $i$ from the diagonal entry to let $\omega_{i-1}=-1$.
		Notice that $\omega_{\hat{i}}=1$, but there does not exist a term with the exponent of $x_{\hat{i}}$ equals to 1 in row $\hat{i}+1$ by Remark~\ref{rem:obs}~(iii).
		
		Next, for some $j \in \{2,\ldots,n\}$, assume that $x_jx_{j-1}^{-1}$ is selected in row $j$. 
		Then we show the selection of terms for multiplication in other rows is unique to obtain $\omega$. 
		The choice of $j$ can be divided into three cases.

		{\bf Case (a).}
		Consider the case when $2\leq j< \hat{i}$.
		In this case, $\omega_{j}=-1$.
		Notice $\omega_{j-1}=-1$, so the exponent of $x_{j-1}$ in monomial in row $j-1$ should be 0. Therefore, $x_{j-2}^{-1}$ is the only selection.
		Likewise, for any $2\leq i\leq j-1$, $x_{i-1}^{-1}$ is selected in row $i$.
		It is direct to check $\omega_1=0$.
		In row $j+1$, $x_{j+1}x_{j}^{-2}$ is the only selection to let $\omega_{j}=-1$.
		Likewise, for any $j+1\leq i\leq \hat{i}$, $x_{i}x_{i-1}^{-2}$ is selected in row $i$.
		Notice $\omega_{\hat{i}}=1$ and $x_{\hat{i}}x_{\hat{i}-1}^{-2}$ is selected in row $\hat{i}$, $x_{\hat{i}+1}$ is the only selection in row $\hat{i}+1$.
		Likewise, for any $\hat{i}+1\leq i\leq n$, $x_{i}$ is selected in row $i$.
		It is direct to check $\omega_n=1$.
		
		Next, we check if these monomials lie on different columns and the sign of the resulting term.
		By \eqref{eq:J4}, in row $i$, $x_ix_{i-1}^{-2}$ appears only in column $i-1$ with a negative coefficient, and $x_{i-1}^{-1}$ appears only in column $i$ with a positive coefficient.
		The coefficients of $x_i$ are $\beta_i$, $\delta_{i,k}\delta_{i+1,k}\beta_{i+1}$ and $\delta_{i-1,k}\delta_{i,k}\beta_{i-1}$.
		Since $x_i$ is selected only when $\hat{i}+1\leq i\leq n$ and $\hat{i}$ is the largest index of the layer that shares phosphatase with other layers, we have $\delta_{i,k}\delta_{i+1,k}=\delta_{i-1,k}\delta_{i,k}=0$.
		Note here that $\beta$s are all positive.
		So, $x_i$ appears only in column $i$ with a positive coefficient.
		To sum up, the selection of monomials for multiplication in other rows to obtain $\omega$ is unique: in row 0, $c_1$ is selected in column 1; in row 1, $x_1x_{0}^{-2}$ is selected in column 0; in row $2\leq i\leq j-1$, $x_{i-1}^{-1}$ is selected in column $i$; in row $j+1\leq i\leq \hat{i}$, $x_{i}x_{i-1}^{-2}$ is selected in column $i-1$; in row $\hat{i}+1\leq i\leq n$, $x_{i}$ is selected in column $i$.
		The only vacant column is $\hat{i}$. This indicates that $x_jx_{j-1}^{-1}$ should appear in column $\hat{i}$.
			Notice that in row $j$ and column $\hat{i}$, the coefficient of $x_jx_{j-1}^{-1}$ is nonzero if and only if $\delta_{j-1,k}\delta_{\hat{i},k} \ne 0$ for some $k$ by \eqref{eq:J4}.
			So case (a) exists if and only if $\iota(j-1) = \iota(\hat{i})$, i.e, the $(j-1)$-th and $\hat{i}$-th layers share the same phosphatase.
			The coefficient of $x_jx_{j-1}^{-1}$ is negative.
		The number of terms with a negative sign is $\hat{i}-j+2$ and the number of reverse order is $\hat{i}-j+1$, so the sign of the resulting term is $(-1)^{\hat{i}-j+2}\cdot(-1)^{\hat{i}-j+1}=-1$.
		
		{\bf Case (b).}
		Consider the case when $j = \hat{i}$.
		In this case, the selection of monomials for multiplication in other rows to obtain $\omega$ is unique too: in row 0, $c_1$ is selected in column 1; in row 1, $x_1x_{0}^{-2}$ is selected in column 0; in row $2\leq i\leq \hat{i}-1$, $x_{i-1}^{-1}$ is selected in column $i$; in row $\hat{i}+1\leq i\leq n$, $x_{i}$ is selected in column $i$.
		Likewise, case (b) exists if and only if $\iota(\hat{i}-1) = \iota(\hat{i})$.
		The number of terms with a negative sign is $2$ and the number of reverse order is $1$, so the sign of the resulting term is $-1$.
		
		{\bf Case (c).}
		Consider the case when $\hat{i}< j\leq n$.
		Notice that $\omega_{j-1}=1$, but there does not exist a term with the exponent of $x_{j-1}$ equals to 2 in row $j-1$.
		
		To sum up, cases (a) and (b) give all possibilities to obtain a term with the exponent $\omega$, and since at least two layers shares a phosphatase, at least one of the two cases always exists. Moreover, the corresponding coefficient is always negative.
		Hence, by Proposition \ref{prop:vtx}, $\omega$ is a negative vertex of $N(f)$.
	\end{proof}
	
	We now use this to prove Theorem~\ref{thm:mainthm2}.
	\begin{proof}[Proof of Theorem~\ref{thm:mainthm2}]
		By Proposition~\ref{lemma:nvtx} there exists a monomial in $p(x)$ that has a negative coefficient and by Proposition~\ref{prop:vtx}, this monomial in fact describes a vertex of the Newton polytope, $\np(p(x)).$ Hence, by Lemma~\ref{lem:negvertex}, for any $\k \in \Rp^{6n}$ and $h \in \Rp^r$, there exists $x^*\in \Rp^{n+1}$ such that  $p_{\k,h}(x^*)<0$.
	\end{proof}
    
	\section{\bf Disconnectedness in $\Sigma_m$ for $n=2$}\label{sec:disconnected}
	
	Here we consider the network for $n=2$ and $r=1$ and show that there are atmost two disconnected regions of multistationarity $\Sigma_{\k,d}$. The network $\mathcal{N}_2$ with shared phosphatase is as follows:
	
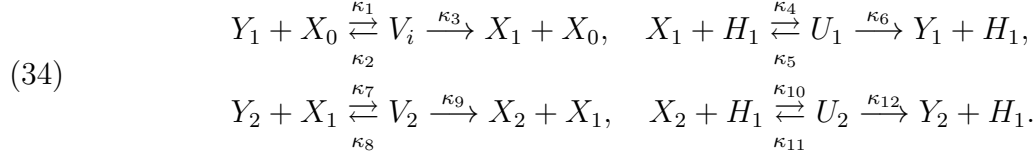
\begin{align}\label{eq:2layernet}
		\begin{split}
			&Y_1+X_{0}  \stackrel{\kappa_{1}}{\underset{\kappa_{2}}{\rightleftarrows}} V_i \stackrel{\kappa_{3}}{\longrightarrow} X_1+X_{0}, \quad
			X_1+H_1  \stackrel{\kappa_{4}}{\underset{\kappa_{5}}{\rightleftarrows}} U_1 \stackrel{\kappa_{6}}{\longrightarrow} Y_1+H_1,  \\
			&Y_2+X_{1}  \stackrel{\kappa_{7}}{\underset{\kappa_{8}}{\rightleftarrows}} V_2 \stackrel{\kappa_{9}}{\longrightarrow} X_2+X_{1}, \quad
			X_2+H_1  \stackrel{\kappa_{10}}{\underset{\kappa_{11}}{\rightleftarrows}} U_2 \stackrel{\kappa_{12}}{\longrightarrow} Y_2+H_1. 
		\end{split}
	\end{align}
	Ordering the species as $H_1,X_0,X_1,X_2,Y_1,Y_2,U_1,U_2,V_1,V_2$ and reactions as labelled in \eqref{eq:2layernet}, we get the stoichiometric matrix $N$ and the reactant matrix $B$ as below. 
	
	{\scriptsize
		\begin{align*}
			N=\begin{bmatrix}
				0 & 0 & 0 & -1 & 1 & 1 & 0 & 0 & 0 & -1 & 1 & 1\\
				-1 & 1 & 1 & 0 & 0 & 0 & 0 & 0 & 0 & 0 & 0 & 0\\
				0 & 0 & 1 & -1 & 1 & 0 & -1 & 1 & 1 & 0 & 0 & 0\\
				0 & 0 & 0 & 0 & 0 & 0 & 0 & 0 & 1 & -1 & 1 & 0\\
				-1 & 1 & 0 & 0 & 0 & 1 & 0 & 0 & 0 & 0 & 0 & 0\\
				0 & 0 & 0 & 0 & 0 & 0 & -1 & 1 & 0 & 0 & 0 & 1\\
				0 & 0 & 0 & 1 & -1 & -1 & 0 & 0 & 0 & 0 & 0 & 0\\
				0 & 0 & 0 & 0 & 0 & 0 & 0 & 0 & 0 & 1 & -1 & -1\\
				1 & -1 & -1 & 0 & 0 & 0 & 0 & 0 & 0 & 0 & 0 & 0\\
				0 & 0 & 0 & 0 & 0 & 0 & 1 & -1 & -1 & 0 & 0 & 0\\
			\end{bmatrix}, \
			B=\begin{bmatrix}
				0 & 0 & 0 & 1 & 0 & 0 & 0 & 0 & 0 & 1 & 0 & 0\\
				1 & 0 & 0 & 0 & 0 & 0 & 0 & 0 & 0 & 0 & 0 & 0\\
				0 & 0 & 0 & 1 & 0 & 0 & 1 & 0 & 0 & 0 & 0 & 0\\
				0 & 0 & 0 & 0 & 0 & 0 & 0 & 0 & 0 & 1 & 0 & 0\\
				1 & 0 & 0 & 0 & 0 & 0 & 0 & 0 & 0 & 0 & 0 & 0\\
				0 & 0 & 0 & 0 & 0 & 0 & 1 & 0 & 0 & 0 & 0 & 0\\
				0 & 0 & 0 & 0 & 1 & 1 & 0 & 0 & 0 & 0 & 0 & 0\\
				0 & 0 & 0 & 0 & 0 & 0 & 0 & 0 & 0 & 0 & 1 & 1\\
				0 & 1 & 1 & 0 & 0 & 0 & 0 & 0 & 0 & 0 & 0 & 0\\
				0 & 0 & 0 & 0 & 0 & 0 & 0 & 1 & 1 & 0 & 0 & 0\\
			\end{bmatrix}.
		\end{align*}
	}
	
	Given the soichiometric matrix $N$ for the above network, we consider a convex pointed polyhedral cone $\ker(N)\cap\Rn^{12}.$ This cone is also known as the flux cone in the literature. It has a minimum generating set of 6 extreme vectors $E_1,\ldots,E_6$ unique upto scaling. We denote the following extreme matrix, $E$, given by these extreme vectors.
	
	{\scriptsize
		\begin{align*}
			E=\begin{bmatrix}
				0 & 0 & 0 & 0 &0 & 0 & 0 & 0 & 0 & 1 & 1 & 0\\
				0 & 0 & 0 & 0 &0 & 0 & 1 & 1 & 0 & 0 & 0 & 0\\
				0 & 0 & 0 & 0 &0 & 0 & 1 & 0 & 1 & 1 & 0 & 1\\
				0 & 0 & 0 & 1 &1 & 0 & 0 & 0 & 0 & 0 & 0 & 0\\
				1 & 1 & 0 & 0 &0 & 0 & 0 & 0 & 0 & 0 & 0 & 0\\
				1 & 0 & 1 & 1 &0 & 1 & 0 & 0 & 0 & 0 & 0 & 0
			\end{bmatrix}^{\top}.
	\end{align*}}

	For a fixed set of reaction rates $\k$, we recall that the steady state variety is given by $V_{\k}$ and we consider the set $V\subset \Rp^{12}\times \Rp^{10}$ given as below:
	\[V:=\{(\k,h,x,y,u,v) \in \Rp^{12}\times \Rp^{10} ~\mid~x \in V_{\k}\}.\]
	Using the matrices $N,B,$ and $E$ we can reparametrise the above set of steady states in terms of convex coordinates. Let $m=(m_1,\ldots,m_{10})\in\Rp^{10}$ and $\ell = (\ell_1,\ldots,\ell_6)\in\Rp^{6}$, consider the following map:
	\begin{equation}\label{eq:parametrization}f: \Rp^{10}\times \Rp^{6}\to V, \qquad (m,\ell)\mapsto (\diag(m^{B_1},\ldots,m^{B_{12}})E\ell,\frac{1}{m})\end{equation} where $B$ is the reactant matrix of \eqref{eq:2layernet} and $B_j$ is the $j$-th column of the matrix $B$ and $m^{B_j}= \prod_{i=1}^{10}m_i^{B_{ij}}.$ 
	The critical function in \eqref{eq:crit} gives a map $q:V \rightarrow \RR$ and using a parameterisation of $V$ in \eqref{eq:parametrization}, we obtain 
	\begin{equation}\label{eq:critconvex}
		\tilde{g}:\Rp^{10}\times \Rp^6 \to \RR
	\end{equation}
	which gives a different expression of the critical function in terms of the convex coordinates $(m,\ell)$. To understand connected sets of regions of parameters in $\Sigma_{\k,d}$ that give multistationarity we need to count the connected components of $\tilde{g}^{-1}(\RR_{<0}).$
	The explicit polynomial expression of $\tilde{g}$ for this network is given by
	\begin{align*}
		\tilde{g}(m,\ell):=	(A_{13}m_1m_3+A_1m_1+A_3m_3+A_0) (\ell_1 + \ell_3)  (\ell_2 + \ell_3)    (\ell_4 + \ell_6)  (\ell_5 + \ell_6)\ell_3  \ell_6
	\end{align*}
	
	where
		\begin{align*}
		A_{13}:=& (m_6 - m_4) (m_{10} m_2 m_7 + m_{10} m_5 m_7 + m_{10} m_5 m_9 + m_{10} m_7 m_9 - 
		m_5 m_8 m_9)\\ 
		A_1:=& m_5 m_9 (m_{10} m_4 m_6 + m_{10} m_6 m_7 + m_{10} m_4 m_8 + m_{10} m_6 m_8 + m_4 m_6 m_8) \\
		A_3:= &m_{10} m_2 m_4 m_6 m_7 + m_{10} m_4 m_5 m_6 m_7 + m_{10} m_2 m_4 m_7 m_8 + 
		m_{10} m_4 m_5 m_7 m_8 + m_{10} m_2 m_6 m_7 m_8 +\\& m_2 m_4 m_6 m_7 m_8 + 
		m_{10} m_5 m_6 m_7 m_8 + m_4 m_5 m_6 m_7 m_8 + m_{10} m_4 m_5 m_6 m_9 + 
		m_{10} m_4 m_6 m_7 m_9 +\\& m_{10} m_4 m_5 m_8 m_9 + m_{10} m_5 m_6 m_8 m_9 + 
		m_4 m_5 m_6 m_8 m_9 + m_{10} m_4 m_7 m_8 m_9 + m_4 m_5 m_7 m_8 m_9 +\\& m_{10} m_6 m_7 m_8 m_9 +
		m_4 m_6 m_7 m_8 m_9\\
		A_0:=&m_5 m_7 m_9 (m_{10} m_4 m_6 + m_{10} m_4 m_8 + m_{10} m_6 m_8 + m_4 m_6 m_8).
	\end{align*}
	
	However, $\tilde{g}$ is negative iff the first factor is negative and hence, we focus on the polynomial $g$ given by:
	
	\begin{equation}
		g:=A_{13}m_1m_3+A_1m_1+A_3m_3+A_0
	\end{equation}
	
	\begin{remark}\label{rem:M}
		Note that $A_1,A_3,$ and $A_0$ has only positive terms and hence, $g$ attains negative values if and only if $A_{13}<0.$ Let $M:=\{(m_2,m_4,m_5,\ldots,m_{10})\in \Rp^8~\mid~A_{13}<0\}.$ Then, $M:=M_1\cup M_2$ where
		{\footnotesize
			\begin{align*}&M_1:=\Big\{(m_2,m_4,m_5,\ldots,m_{10})\in \Rp^8~\mid~ m_6 < m_4 \text{ and } 
				m_{8} < \frac{m_{10} m_2 m_7 + m_{10} m_5 m_7 + m_{10} m_5 m_9 + m_{10} m_7 m_9}{m_5 m_9}\Big\},\\
				&M_2:=\Big\{(m_2,m_4,m_5,\ldots,m_{10})\in \Rp^8~\mid~ m_4 < m_6 \text{ and } 
				m_{8} > \frac{m_{10} m_2 m_7 + m_{10} m_5 m_7 + m_{10} m_5 m_9 + m_{10} m_7 m_9}{m_5 m_9}\Big\}.
		\end{align*}}
		From the above expressions it is easy to check that $M_1$ and $M_2$ are disconnected.
	\end{remark}
	
	We will now prove the main result of this section.
	
	\begin{theorem}\label{thm:disconnected}
		There are at most two connected components of multistationarity region in $\Sigma_{\k,d}.$
	\end{theorem}
	
\begin{proof}
	First we note that when $m_4=m_6$, we have that the polynomial $g$ is a positive polynomial. Next we will show that $g^{-1}(\RR_{<0})$ has two connected components. As noted in Remark~\ref{rem:M}, $A_{13}<0$ is a necessary condition for $g<0$. We first consider $g$ as a polynomial only in $m_1,m_3$ and denote it by $g_0$. For any point in $M$, the newton polytope of $g_0$ has exactly one negative vertex and hence, admits a separating hyperplane. By \cite[Theorem 3.6]{telek1}, we then have that for any point in $M$ the closure of $g_0^{-1}(\RR_{<0})$ is given by $g_0^{-1}(\RR_{\leq 0})$ and $g_0^{-1}(\RR_{<0})$ has a unique connected component. However, since, $M$ is disconnected into two distinct components, $g^{-1}(\RR_{<0})$ has two disconnected components. These components are explicitly given by:
	\[D:=\Big\{m\in \Rp^{10}~\mid~ A_{13}<0, ~m_1>\frac{A_3}{-A_{13}}, \text{ and }m_3>\frac{A_0+A_1 m_1}{-A_3+A_{13}m_1}\Big\}.\]
	
	Finally, by \cite[Theorem 4]{ConnectivityPaper}, the multistationarity regions in full parameter space $\Sigma_{\k,d}$ has at most two connected components.
\end{proof}

	While the above result only gives the upper bound, we conclude this section with a conjecture on the number of connected regions of parameter space for multistationarity over all the parameters including reaction rates and the total concentrations.
	
	\begin{conjecture}\label{conj}
		Parameter region of multistationarity is disconnected in the parameter space~$\Sigma_{\k,d}$.
	\end{conjecture}

	\section{\bf Outlook}
	
	In this article, we considered the family of cascade systems of Goldbeter-Koshland loops. These systems can have multiple sites (indexed by $n$) available for phosphorylation and such a system admits multistationarity only if at least two of the sites share a phosphatase for dephosphorylation \cite{EF2012, MS2018}. For these networks we explore the topology of parameter regions that allow for multistationarity. In particular, we are interested in the number of path connected regions in the parameter space. 
	
	One of the main results show that for arbitrary $n\geq 2$, in the space of reaction rate constants, $\Sk$, the set of reaction rates for which there exists some stoichiometric compatibility class that admits multistationarity is path connected. In fact, we show that it is all of $\Sk$. We prove this by studying the structure of the critical polynomial of this network and showing that one of the vertices of its newton polytope always have a negative coefficient. This gives a lower bound on the number of connected components for parameter regions for multistationarity in the full parameter space $\Skd.$ Finally, we only focus on the network with $n=2$ sites for phosphorylation and show that for this network the upper bound for the number of connected components in $\Skd$ is 2. In fact, we conjecture that the actual number in this case is same as the upper bound.
	
	The existing methods in the literature give a sufficient condition for path-connectivity and they have been used to explore this for important families networks in the literature \cite{FKWY,FKWY2,kaihnsa_connectivity_2024,ConnectivityPaper}. The discrepancy in the lower and upper bound for the number of connected components for this networks highlights the need for developing new methods. At the same time, this network presents itself as an interesting test case for developing those techniques. Besides the stated conjecture, it is also unclear if there are certain sufficient structural conditions on the network that ensures the upper bound is more than 1 or how this number evolves for this family of networks when $n>2$.
	
\vspace{1cm}
	
	\paragraph{\bf Acknowledgement} This work emerged from a group project during the MSRI-MPI Leipzig Summer Graduate School on Algebraic Methods for Biochemical Reaction Networks.  NK has
	been funded by the European Union under the Grant Agreement no. 101044561, POSALG. Views and opinions
	expressed are those of the authors only and do not necessarily reflect those of the European Union or European
	Research Council (ERC). Neither the European Union nor ERC can be held responsible for them.

	\vspace{1cm}
	
	\noindent
	\footnotesize {\bf Authors' addresses:}
	
	\medskip
	
	\noindent Nidhi Kaihnsa,
	\  University of Copenhagen \hfill  {\tt nidhi@math.ku.dk}
	
	\noindent Kaizhang Wang,
	\  Peking University \hfill  {\tt wangkz@stu.pku.edu.cn}

\end{document}